\pdfoutput=1
\documentclass[10pt,a4paper,reqno]{amsart}
\usepackage{amssymb,amsmath,amsthm}
\usepackage{hyperref}
\usepackage{graphicx}
\usepackage{enumerate}
\usepackage{url}
\usepackage{booktabs}
\usepackage{subfigure}
\usepackage{siunitx}
\usepackage[foot]{amsaddr}
\usepackage[T1]{fontenc}
\usepackage{mathtools}
\usepackage{tabularx}

\usepackage{etoolbox}
\makeatletter
\patchcmd{\@maketitle}
  {\ifx\@empty\@dedicatory}
  {\ifx\@empty\@date \else {\vskip3ex \centering\footnotesize\@date\par\vskip1ex}\fi
   \ifx\@empty\@dedicatory}
  {}{}
\patchcmd{\@adminfootnotes}
  {\ifx\@empty\@date\else \@footnotetext{\@setdate}\fi}
  {}{}{}
\makeatother

\newcounter{rownumber}[figure] 
\numberwithin{equation}{section}

\theoremstyle{plain}%
     
\newtheorem{prop}{Proposition}[]
\newtheorem{thm}{Theorem}[]
\newtheorem{lem}{Lemma}[]

\newcounter{example}

\newcommand{\1}{{\textrm{1} \kern -.41em \textrm{1} }}

\newcommand\E{\mathbb{E}}

\renewcommand\P{\mathbb{P}}
\newcommand\NP{\mathbb{N}\mathbb{P}}

\begin{document}
\title[Reconstructing binary matrices]{Reconstructing binary matrices under window constraints from their row and column sums}
\author{Andreas Alpers and Peter Gritzmann}
\address{Zentrum Mathematik, Technische Universit\"at M\"unchen, D-85747 Garching bei M\"unchen, Germany}
\email{alpers@ma.tum.de, gritzman@tum.de}
\thanks{The authors gratefully acknowledge support through the German Research Foundation Grant GR 993/10-2 and the European COST Network MP1207.}

\begin{abstract}
The present paper deals with the discrete inverse problem of reconstructing binary matrices from their row and column sums under additional constraints on the number and pattern of entries in specified minors. While the classical consistency and reconstruction problems for two directions in discrete tomography can be solved in polynomial time, it turns out that these window constraints cause various unexpected complexity jumps back and forth from polynomial-time solvability to $\mathbb{N}\mathbb{P}$-hardness.
\end{abstract}

\maketitle
\section{Introduction} 
The problem of reconstructing binary matrices from their row and column sums  is a classical inverse task in combinatorics; see \cite{ryser57}.
Even though the term was introduced much later it can be seen as a root of {\em discrete tomography}; see \cite{sksbko-93, ksbsko-95, gg97, gritzmann97,batenburgnature}, and the collected editions \cite{kubaherman1}, \cite{kubaherman2}.
Of particular relevance for the present paper are two well-known observations: (1) The question of consistency of the data i.e., the question whether there exists a matrix whose row and column sums coincide with the given data, can be solved in polynomial time.
(2) Typically, its row and column sums do not determine the underlying matrix uniquely.  See  \cite{ryser57, fishburn91, aharoni97, glw11}  for characterizations of the `rare' cases of uniqueness.

In the present paper we address the second issue by adding additional window constraints, specifying (or giving bounds on) how many points there are in certain minors  of the matrix. These constraints come up naturally in {\em dynamic discrete tomography}
\cite{agdynamic}. An application of particular relevance in physics is that of particle tracking; see \cite{agms-15}, \cite{glidingarc-15}. Here, the positions of particles over time are to be reconstructed from two (sometimes more) high speed camera images. Window constraints are a natural way of modeling additional physical information for instance on the speed of the particles; see \cite{agdynamic} and the background literature given there for further information. 

Here we are taking a complexity theoretical view commemorating Observation (1). In fact, when adding various kinds of window constraints, we are interested in the boundary between polynomial-time solvability and $\NP$-hardness. Intuitively speaking, we ask which kinds of additional constraints can be added without imposing a significant extra computational effort. Or, phrased differently, what is the computational price to pay for reducing the number of solutions by utilizing additional window information. 
We will focus on the effect of three different parameters: $k$ corresponds to the {\em size} of the window, $\nu$ is the {\em number of $1$'s} in the nonzero minors, and $t$ specifies the  allowed positions of $1$'s in the windows, referred to as the {\em pattern} of the window. The choice for these parameters will specify the given problem \textsc{Rec}$(k,\nu,t),$ which is formally introduced in Section~\ref{sec-results}. Hence $k,$ $\nu,$ and $t$ are given beforehand (i.e., are not part of the input). As it will turn out, the problem exhibits various unexpected complexity jumps.

Omitting technical details (which are all given in Section \ref{sec-results}) some of these jumps can be summarized as follows; see Table \ref{table:sumtable1}. 

For $k=1$ the problems are in $\P$ regardless on how the other parameters are set; see Theorem~\ref{thm:tract}(\ref{thm:tract:i}). For $k\ge 2$ it depends on $\nu$ and $t$ whether the problems are in $\P$ or $\NP$-hard; see
Theorems~\ref{thm:tract}(\ref{thm:tract:ii}),(\ref{thm:tract:iii}) and~\ref{thm:NP}. For $k\ge 2$, some values of $\nu$ render the problems tractable while others make them $\NP$-hard; see Theorems~\ref{thm:tract}(\ref{thm:tract:ii}) and~\ref{thm:NP}(\ref{thm:NP:ii}) even if the patterns are not restricted at all. Adding a pattern constraint may turn an otherwise $\NP$-hard problem into a polynomial time solvable problem; see Theorems~\ref{thm:NP}(\ref{thm:NP:ii}) and~\ref{thm:tract}(\ref{thm:tract:iii}). The reverse complexity jump, however, can also be observed; see Theorems~\ref{thm:tract}(\ref{thm:tract:ii}) and~\ref{thm:NP}(\ref{thm:NP:i}). 

The present paper is organized as follows: Section \ref{sec-results} will introduce our notation and state the main results. The proofs of our tractability results are given in Section \ref{sec-tractability} while the $\NP$-hardness results are proved in Section \ref{sec-intractability}. Section~\ref{sect:finalremarks} contains some final remarks. In particular, it provides a (potentially also quite interesting) extension of our main problem and collects related implications of our main results.

\section{Notation and Main Results}\label{sec-results}
We begin with some standard notation.

Let $\mathbb{Z},$ $\mathbb{N},$ and $\mathbb{N}_0$ denote the set of integers, natural numbers, and non-negative integers, respectively. For $k \in \mathbb{N}$ set $k\mathbb{N}_0:=\{ki:i\in \mathbb{N}_0\}$, $k\mathbb{N}_0+1:=\{ki+1:i\in\mathbb{N}_0\},$  $[k]:=\{1,\dots,k\}$, and $[k]_0:=[k]\cup\{0\}.$ 
The support of a vector $x:=(\xi_1,\dots,\xi_d)^T \in \mathbb{Z}^d$ is defined as  $\textnormal{supp}(x):=\{i:\xi_i\neq0\}$. 
With $\1$ we denote the all-ones vector of the corresponding dimension. The cardinality of a finite set $F\subseteq \mathbb{Z}^d$ is denoted by~$|F|.$ We will use the notational convention that specific settings of variables and parameters are signified by a superscript $^*.$

In this paper we often refer to points (or variables) $\xi_{p,q}$ where $(p,q)$ are points of the integer grid~$\mathbb{Z}^2$. Then, $p$ and $q$ denote the $x$- and $y$-coordinates of $(p,q)$, respectively. In the grid $[m]\times[n],$ the  set $\{i\}\times[n]$ and $[m]\times\{j\}$ is called \emph{column~$i$} and \emph{row~$j,$} respectively.  For $k\in\mathbb{N}$ we refer to  $\{(i-1)k+1,\dots,ik\}\times[n]$ and $[m]\times\{(j-1)k+1,\dots,jk\}$ as \emph{vertical strip~$i$ (of width~$k$)} and \emph{horizontal strip~$j$ (of width~$k$)}, respectively. Of course, the standard matrix notation can be obtained from our notation by the coordinate transformation
\[
\left(\begin{array}{ccc} \xi_{1,n} & \cdots & \xi_{m,n}\\ \vdots & & \vdots\\ \xi_{1,1} & \cdots & \xi_{m,1} \end{array}\right) \mapsto \left(\begin{array}{ccc} \xi_{1,1} & \cdots & \xi_{1,n}\\ \vdots & & \vdots\\ \xi_{m,1} & \cdots & \xi_{m,n} \end{array}\right).
\]

Generally, any subset of the $[m]\times[n]$ grid is called a \emph{window}. Windows of the form $([a,b]\times[c,d])\cap\mathbb{Z}^2,$ with $a,b,c,d\in\mathbb{Z}$ and $a\leq b,$ $c\leq d,$ are called \emph{boxes.} Defining for any $k\in\mathbb{N}$ and $m,n\in k\mathbb{N}$ the set of \emph{(lower-left) corner points} $C(m,n,k):=([m]\times [n])\cap(k\mathbb{N}_0+1)^2,$ we call any box $B_k(i,j):=(i,j)+[k-1]^2$ with $(i,j)\in C(m,n,k)$ a \emph{block}; see Figure~\ref{fig:ill} for an illustration. (Here we depict the structure both as point set and as pixels. In the following we restrict the figures to pixel images, which we find more intuitive.) The blocks form a partition of $[m]\times[n],$ i.e., $\bigcup_{(i,j)\in C(m,n,k)}B_k(i,j)=[m]\times[n].$  In the main part of this paper, we consider such non-overlapping blocks, which play also a role in super-resolution imaging \cite{agsuperresolution}. In Section~\ref{sect:finalremarks} we consider also other windows, which may be positioned at other places than those defined by $C(m,n,k).$

\begin{figure}[htb] 
\centering
\subfigure[]{\includegraphics[width=0.19\textwidth]{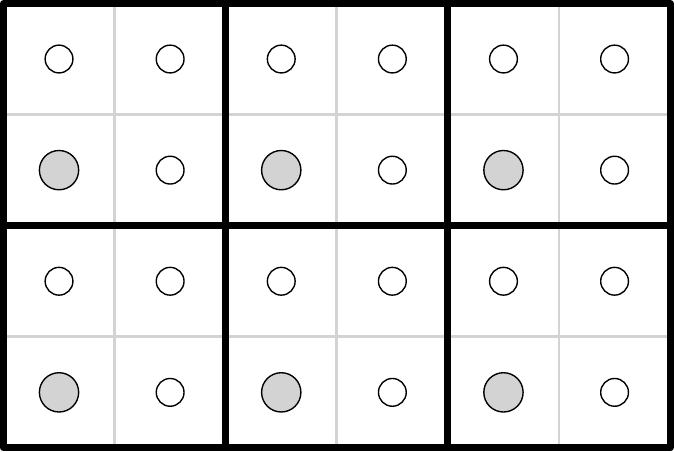}}\hspace*{4ex}
\subfigure[]{\includegraphics[width=0.19\textwidth]{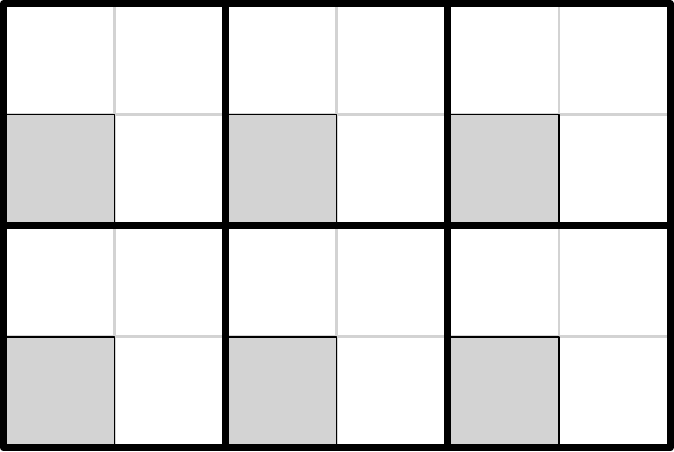}}
\caption{(a) Point set, (b) the same structure depicted as pixels showing the six blocks $B_k(i,j),$ $(i,j)\in C(m,n,k)$ for $m=6,$ $n=4,$ and $k=2.$ The elements of $C(m,n,k)$ are highlighted.}\label{fig:ill}
\end{figure}

Next we introduce three patterns that we will study in detail since they exhibit already the general complexity jump behavior we are particular interested in. 

The first pattern $P(k,0)$ is unconstrained, i.e., does not pose any additional restrictions on the positions of $1$'s. The second pattern $P(k,1)$ forces all elements in the $k\times k$ block to be $0$ except possibly for the two entries in the lower-left and upper-right corner. The third pattern $P(k,2)$ excludes all patterns that admit more than one~$1$ in each row of the $k\times k$ block (see also Figure~\ref{fig:C1}). Here are the formal definitions.

For $k\in \mathbb{N}$ let $2^{[k-1]_0^2}$ denote the power set of $[k-1]_0^2.$ Then we set
\begin{align*}
P(k,0)&:=2^{[k-1]_0^2},\\
P(k,1)&:=\{\{(0,0)\},\{(k-1,k-1)\}\},\\
P(k,2)&:=\{M\in2^{[k-1]_0^2} \::\: |M\cap\left([k-1]_0\times\{j\}\right)|\leq 1 \textnormal{ for all } j\in [k]_0.\}
\end{align*} 

Further, for $(i,j)\in C(m,n,2)$ and $x=(\xi_{i,j})_{i\in[m],j\in[n]}$ we set \[\textnormal{pat}_k(x,i,j):=\{(p,q)\in B_k(i,j)\::\: \xi_{p,q}\neq0\}-(i,j).\]

For $k,\nu\in\mathbb{N}$ and $t\in\{0,1,2\}$ we define now the following problem \textsc{Rec}$(k,\nu,t).$

\begin{center}
\begin{minipage}{0.95\textwidth}
\textsc{Rec}$(k,\nu,t)$\\[-3ex]

\hspace*{1ex}\begin{minipage}{0.98\textwidth}
\begin{alignat*}{6}
&\textnormal{Instance:}\quad&& \omit \rlap{$\displaystyle m,n \in k\mathbb{N},$}\\
&&&\omit \rlap{$\displaystyle r_1,\dots,r_n \in \mathbb{N}_0,$}&&&&&& \textnormal{(row sum measurements)} \\
&&&\omit \rlap{$\displaystyle c_1,\dots,c_m \in \mathbb{N}_0,$}&&&&&& \textnormal{(column sum measurements)} \\
&&&\omit \rlap{$v(i,j)\in\{0,\nu\},$}   &&&& (i,j)\in C(m,n,k), &&  \textnormal{(block measurements)}\\[1.2ex]
&\textnormal{Task:}\quad&& \omit \rlap{Find  $\xi_{p,q}\in\{0,1\},$ \:$(p,q)\in[m]\times[n],$ with }\\
\displaystyle &&&\sum_{\mathclap{p\in[m]}} \xi_{p,q}&&=r_q,&&   q\in[n], &&\textnormal{(row sums)} \\ 
\displaystyle &&&\sum_{q\in[n]} \xi_{p,q}&&=c_p,&&  p \in [m], &&\textnormal{(column sums)}\\  
\displaystyle &&&\sum_{\mathclap{(p,q)\in B_k(i,j)}} \xi_{p,q}&& \leq v(i,j),&&  (i,j)\in C(m,n,k), \quad &&\textnormal{(block constraints)}\\
\displaystyle &&&\textnormal{pat}_k(x,i,j)&&\in P(k,t),\qquad&&  (i,j)\in C(m,n,k),\qquad &&\textnormal{(pattern constraints),}\\[1.2ex]
&&& \omit \rlap{or decide that no such solution exists.}\\
\end{alignat*}
\end{minipage}
\end{minipage}
\end{center}

In other words, we ask for $0/1$-solutions that satisfy given row and column sums, \emph{block constraints} of the form $\sum_{(p,q)\in B_k(i,j)}\xi_{i,j}\leq v(i,j)$ with given $v(i,j)\in\{0,\nu\},$ and \emph{pattern constraints} that restrict the potential locations of the 1's in each block. We remark that the $v(i,j),$ $(i,j)\in C(m,n,k),$ can be viewed as some part of prior knowledge. In particle tracking \cite{glidingarc-15, agms-15}, for instance, the $v(i,j)$ may reflect prior knowledge about physically meaningful particle trajectories; see also \cite{agdynamic}.

Our main results show that the computational complexity of $\textsc{Rec}(k,\nu,t)$ may change drastically when~$k,$~$\nu,$ or~$t$ is varied. 

\begin{thm} \label{thm:tract} \hfill
\vspace*{-1ex}\begin{enumerate}[(i)]
\itemsep1ex 
\item \textsc{Rec}$(1,\nu,t)\in\mathbb{P}$ for any $\nu\in\mathbb{N}$ and $t\in\{0,1,2\}.$ \label{thm:tract:i}
\item \textsc{Rec}$(k,1,0)\in\mathbb{P}$ for any $k\geq2.$ \label{thm:tract:ii}
\item \textsc{Rec}$(k,\nu,2)\in\mathbb{P}$ for any $k\geq2$ and $\nu\geq k.$ \label{thm:tract:iii}
\end{enumerate}
\end{thm}

\begin{thm} \label{thm:NP} \hfill
\vspace*{-1ex}\begin{enumerate}[(i)]
\itemsep1ex 
\item \textsc{Rec}$(k,1,1)\in\mathbb{N}\mathbb{P}$-hard for any $k\geq2.$ \label{thm:NP:i}
\item \textsc{Rec}$(k,2,0)\in\mathbb{N}\mathbb{P}$-hard for any $k\geq2.$ \label{thm:NP:ii}
\end{enumerate}
\end{thm}

The most notable changes are summarized in Table~\ref{table:sumtable1}.  Some of these changes may at first glance seem somewhat counterintuitive. For instance, restricting the solution space via pattern constraints turns the $\mathbb{N}\mathbb{P}$-hard problem \textsc{Rec}$(k,2,0)$ $(k\geq 2),$ into the polynomial time solvable problem \textsc{Rec}$(k,2,2);$ Figure~\ref{fig:C1} depicts the possible types of blocks in \textsc{Rec}$(k,2,2).$  Conversely, additional pattern constraints convert the tractable problem \textsc{Rec}$(k,1,0)$ into the $\mathbb{N}\mathbb{P}$-hard problem  \textsc{Rec}$(k,1,1)$ $(k\geq 2).$

\begin{table}[htb]
\begin{tabular}{ccc}\toprule
& $\mathbb{P}$ & $\mathbb{N}\mathbb{P}$-hard\\\midrule
varying $k$ & \textsc{Rec}$(1,2,0)$ & \textsc{Rec}$(k,2,0)$ \\
& \textsc{Rec}$(1,1,1)$ & \textsc{Rec}$(k,1,1)$\\[2ex]
varying $\nu$ & \textsc{Rec}$(k,1,0),$ & \textsc{Rec}$(k,2,0)$\\[2ex]
varying $t$ & \textsc{Rec}$(k,2,2)$ & \textsc{Rec}$(k,2,0)$ \\
&  \textsc{Rec}$(k,1,0)$ &\textsc{Rec}$(k,1,1)$\\\bottomrule\\[-.1ex]
\end{tabular}
\caption{Computational complexity of $\textsc{Rec}(1,\nu,t)$ and $\textsc{Rec}(k,\nu,t)$ for $k\geq 2$ under change of a single parameter.}\label{table:sumtable1}
\end{table}

\begin{figure}[htb] 
\centering
\includegraphics[width=0.5\textwidth]{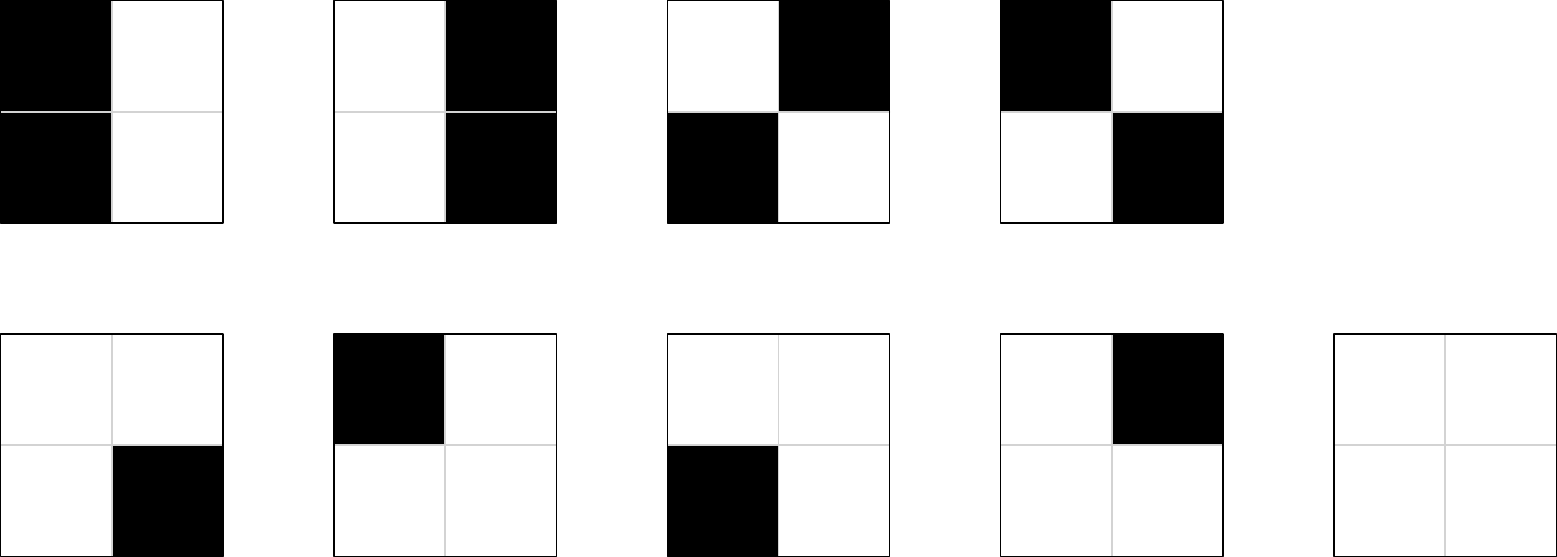}
\caption{The 9 possible types of $1$-entries (black pixels) in each block of any solution of a $\textsc{Rec}(2,2,2)$ instance.}\label{fig:C1}
\end{figure}

\section{Tractability results: Polynomial-time solvability}\label{sec-tractability}

This section contains the proofs of our tractability results stated in Theorem~\ref{thm:tract}.

Clearly, a polynomial-time algorithm for \textsc{Rec}$(1,\nu,t)$ can be given along the classical lines for deciding whether given row and column sums of a binary matrix are consistent. In the following we give one of the various arguments.

\begin{proof}[Proof of Theorem~\ref{thm:tract}(\ref{thm:tract:i})]
The problem obviously reduces to that of reconstructing a $0/1$-matrix from given row and column sums with some of the matrix entries fixed to $0.$ 
It is well known that the coefficient matrix is totally unimodular, see, e.g.,  \cite{gritzmann97, aharoni97}, \cite[Sect.~19.3]{schrijver}. The problem can thus be solved efficiently, e.g., by linear programming. 
\end{proof}

Next we turn to the proof of Theorem~\ref{thm:tract}(\ref{thm:tract:ii}). For $(i,j)\in I\subseteq C(m,n,k),$ let 
\[
\sigma_i(j):=|(\{i\}\times[j])\cap I|\qquad \textnormal{and}\qquad \rho_j(i):=|([i]\times\{j\})\cap I|
\] 
denote the number of blocks~$B_k(a,b)$ with $(a,b)\in I,$ $b\leq j$ and  $a\leq i,$  that lie in the same vertical and horizontal strip, respectively, as the block $B_k(i,j).$
Moreover, for $I\subseteq C(m,n,k)$ we set 
\begin{align*}
G(I)&:=\bigcup_{(i,j)\in I}B_k(i,j)\subseteq [m]\times[n],\\
\Pi_x(I)&:=\{i\in[m]:\exists j\in[n]:(i,j)\in I\},\\
\Pi_y(I)&:=\{j \in [n]:\exists i \in [n]:(i,j)\in I\}.
\end{align*}
Note that $\Pi_x(I)$ and $\Pi_y(I)$ denotes the projection of~$I$ onto the first and second coordinate, respectively.

We use a (slightly generalized version of a) result of \cite{agsuperresolution} on the following problem.

\begin{center}
\begin{minipage}{0.95\textwidth}
\textsc{DR}$(1)$\\[-3ex]

\hspace*{2ex}\begin{minipage}{0.5\textwidth}
\begin{alignat*}{6}
&\textnormal{Instance:}\quad&& \omit \rlap{$\displaystyle m,n \in k\mathbb{N},$}\\
&&&\omit \rlap{$\displaystyle I\subseteq C(m,n,k),$}&&\hspace*{11ex}&&&\hspace*{8ex}& \textnormal{(a set of corner points)}\\
&&&\omit \rlap{$\displaystyle r_{j+l}\in \mathbb{N}_0,$}&&&& j\in\Pi_y(I), \:\: l\in [k-1]_0,&& \textnormal{(row sum measurem.)}  \\
&&&\omit \rlap{$\displaystyle c_{i+l}\in \mathbb{N}_0,$}&&&& i\in\Pi_x(I), \:\: l\in [k-1]_0,&& \textnormal{(column sum measurem.)}\\[1.2ex]
&\textnormal{Task:}\quad&& \omit \rlap{Find  $\xi_{p,q}\in\{0,1\}, \:(p,q)\in G(I)$ with }\\
\displaystyle &&&\sum_{\mathclap{p:(p,j)\in G(I)}}\xi_{p,j+l}&&=r_{j+l},&&  j\in\Pi_y(I), \:\: l\in [k-1]_0,&&\textnormal{(row sums)} \\ 
\displaystyle &&&\sum_{\mathclap{q:(i,q)\in G(I)}}\xi_{i+l,q}&&=c_{i+l},&&  i\in\Pi_x(I), \:\: l\in [k-1]_0,&&\textnormal{(column sums)}\\  \displaystyle &&&\sum_{\mathclap{(p,q)\in B_k(i,j)}}\xi_{p,q}&&=1,&\:\:& (i,j)\in I, &&\textnormal{(block constraints),}\\
&&& \omit \rlap{or decide that no such solution exists.}
\end{alignat*}
\end{minipage}
\end{minipage}
\end{center}

\begin{prop}[\cite{agsuperresolution}] \label{prop:mono}\hfill
\begin{enumerate}[(i)]
\item \label{mono2} An instance $\mathcal{I}$ of \textsc{DR}$(1)$ is feasible if, and only if, for every $(i,j) \in I$ we have
\[\sum_{l=0}^{k-1}r_{j+l}=\rho_j(m) \quad \textnormal{ and } \quad \sum_{l=0}^{k-1}c_{i+l}=\sigma_i(n).\]
\item \label{mono1} \textsc{DR}$(1)\in\mathbb{P}.$ 
\end{enumerate}
\end{prop}

As a service to the reader we remark that a solution $\xi^*_{p,q},$ $(p,q)\in G(I),$ for a given instance of \textsc{DR}$(1)$ is obtained by setting for every $(i,j) \in I$ and $(p,q)\in B_k(i,j):$ 
\begin{equation} \label{eq:sol1}
\begin{aligned}
a_{i,j}&:=i+\min\{l\in\{0,1\}:\sigma_{i}(j)\leq \sum_{h=0}^{l}c_{i+h}\},\\
b_{i,j}&:=j+\min\{l\in\{0,1\}:\rho_{j}(i)\leq \sum_{h=0}^{l}r_{j+h}\},\\
\xi^*_{p,q}&:=\left\{\begin{array}{lll}1&:& (p,q)=(a_{i,j},b_{i,j}),\\ 0&:&\textnormal{otherwise.} \end{array}\right.
\end{aligned} 
\end{equation}
An illustration is given in Figure~\ref{fig:Example1} (taken from \cite{agsuperresolution}).

\begin{figure}[htb]
\centering
\includegraphics[width=0.45\textwidth]{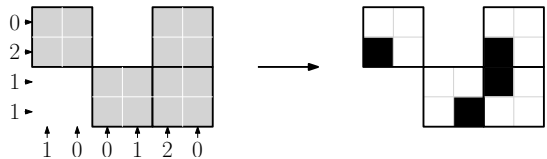}
\caption{(\cite{agsuperresolution}) Illustration of \textsc{DR}(1). (Left) Row and column sums and blocks $B_k(i,j)$ with $(i,j)\in I$ in gray color. (Right) Solution defined by~\eqref{eq:sol1}.}\label{fig:Example1}
\end{figure}

Now we provide a polynomial-time algorithm for \textsc{Rec}$(k,1,0).$

\begin{proof}[Proof of Theorem~\ref{thm:tract}(\ref{thm:tract:ii})]
Let $\mathcal{I}=(m,n,r_1,\dots,r_n,c_1,\dots,c_m,v(1,1),\dots,v(m-k+1,n-k+1))$ denote an instance of \textsc{Rec}$(k,1,0).$ We proceed in two steps.

First, we solve the following  reconstruction-from-row-and-column-sums  instance of finding $\eta_{i,j}\in\{0,1\},$ $(i,j)\in C(m,n,k),$ satisfying the constraints
\begin{equation}\label{eq:rowcolsums}
\begin{aligned}
 \sum_{i:(i,j)\in I} \eta_{i,j}&=\sum_{l=0}^{k-1}r_{j+l},  \qquad &&j\in[n]\cap(k\mathbb{N}_0+1),\\
 \sum_{j:(i,j)\in I} \eta_{i,j}&=\sum_{l=0}^{k-1}c_{i+l},  &&i\in [m]\cap(k\mathbb{N}_0+1),\\
 \eta_{i,j}&=0, &&(i,j)\in \{(a,b)\in C(m,n,k):v(a,b)=0\}.
\end{aligned}
\end{equation} If no solution exists, we report infeasibility of $\mathcal{I}.$ The idea behind solving~\eqref{eq:rowcolsums} first is that in this way we locate boxes that we want to contain a~1.

In a second step, if \eqref{eq:rowcolsums} is feasible, we invoke \textsc{DR}$(1)$ from Proposition~\ref{prop:mono} for the instance $\mathcal{I}'$ where \[\mathcal{I'}:=(m,n, I,r_1,\dots,r_{n},c_1,\dots,c_{m})\] and
\begin{equation} \label{eq:Idef}
 \qquad  I:=\{(i,j)\in C(m,n,k):\eta^*_{i,j}=1\}.\end{equation} If there is no solution we report infeasibility of $\mathcal{I};$ otherwise, we return the solution of \textsc{DR}$(1).$ It remains to be shown that this solves~$\mathcal{I}$ correctly. 

To this end, suppose that $\mathcal{I}$ has a solution $\xi^*_{i,j},$ $(i,j)\in [m]\times[n].$ Then,~\eqref{eq:rowcolsums} is also feasible (a solution is given by $\eta^*_{i,j}:=\sum_{(p,q)\in B_k(i,j)}\xi^*_{p,q},$ $(i,j)\in C(m,n,k)$) and $\xi_{i,j}^*$ satisfies the corresponding instance of \textsc{DR}$(1).$  Conversely, let $\eta_{i,j}^*,$ $(i,j)\in C(m,n,k),$ denote a solution to~\eqref{eq:rowcolsums} and let $I$ be defined as in~\eqref{eq:Idef}. For every $(i,j) \in I$ we have clearly
\[\sum_{l=0}^{k-1}r_{j+l}=\rho_j(m) \quad \textnormal{ and } \quad \sum_{l=0}^{k-1}c_{i+l}=\sigma_i(n).\]  By Proposition~\ref{prop:mono} there is thus a solution to \textsc{DR}$(1).$ This solution satisfies all row and column sum constraints of~$\mathcal{I}$ by definition. Further since $\{(i,j)\in C(m,n,k): v(i,j)=0\}\cap I=\emptyset$ by definition of~\eqref{eq:rowcolsums}, there are 1's in the solution only in those blocks $B_k(i,j),$ $(i,j)\in C(m,n,k),$ for which $v(i,j)=1$ holds.  The block constraints are thus also satisfied, which proves the claim.

The integer linear program~\eqref{eq:rowcolsums} can be solved in polynomial time as the coefficient matrix is well known to be totally unimodular; see, e.g., \cite{gritzmann97, aharoni97}, \cite[Sect.~19.3]{schrijver}. The set $I$ in~\eqref{eq:Idef} is determined in polynomial time, and \textsc{DR}$(1)\in\mathbb{P}$ by Proposition~\ref{prop:mono}(\ref{mono1}). Hence, \textsc{Rec}$(k,1,0)\in \mathbb{P}.$ 
\end{proof}

Next we turn to the proof of Theorem~\ref{thm:tract}(\ref{thm:tract:iii}).
It relies on the following result. 

\begin{lem} \label{prop:totallyunimodular}
Let $A=(a_1,\dots,a_{m_1})^T\in  \{0,1\}^{m_1\times n_1}$ be the node-edge incidence matrix of a (simple) bipartite graph, and let $H=(h_1,\dots,h_{m_2})^T$ denote a binary $m_2\times n_1$ matrix with the following properties.
\begin{enumerate}[(i)]
\item For every $l \in [m_2]$ there exists an index $i\in[m_1]$ with $\text{supp}(h_l)\subseteq \text{supp}(a_i)$; \label{lem1:prop2}
\item The supports of any two distinct vectors $h_i$ and $h_l$ of $H$ with $\text{supp}(h_i)\cup\text{supp}(h_l)\subseteq \text{supp}(a_k)$ for some $k\in[m_1]$ do not intersect. \label{lem1:prop3}
\end{enumerate}
Then, $\left(\begin{array}{l}A\\H\end{array}\right)$ is totally unimodular.
\end{lem}

\begin{proof} It suffices to prove the assertion under the additional assumption that 
\begin{equation} \label{eq:h}
|\textnormal{supp}(h_l)|\geq2, \qquad \textnormal{for all }l\in[m_2].
\end{equation}
All vectors $h_l$ whose support is a singleton can be added later since appending any subset of rows of the $m_1\times n_1$-identity matrix to a totally unimodular matrix yields again a totally unimodular matrix.

As the matrix $A$ is the node-edge incidence matrix of a bipartite graph, let $(R_1,R_2)$ denote a corresponding partition of (the indices of) the vertices of the graph i.e., a partition of $[m_1]$. Clearly, 
for every  $i\neq l$ with $(i,l)\subseteq (R_1\times R_1)\cup(R_2\times R_2)$  we have
\begin{equation}
\label{eq:prooflem1}
\text{supp}(a_i)\cap\text{supp}(a_l)=\emptyset.
\end{equation} 

It suffices to prove that every collection of rows of 
\[M:=\left(\begin{array}{l}A\\H\end{array}\right)\] 
can be split into two parts such that the sum of the rows in one part minus the sum of the rows in the other part is a vector with entries in $\{0,\pm 1\}$; see e.g. \cite[Theorem~19.3]{schrijver}.
Hence, let $I\in [m_1]$ and $L\in [m_2]$ denote such a collection of row indices of~$A$ and~$H,$ respectively. For  $R \subseteq [m_1]$ we set
\[
L(R):=\{l \in L \::\: \exists i\in R \text{ with } \text{supp}(h_l)\subseteq \text{supp}(a_i)\}.\] 
Of course, the sets $(R_1\cap I)$ and $(R_2\cap I)$ form a partition of $I,$ i.e.,
\begin{equation}\label{eq:partition1}
I=(R_1\cap I) \:\:\cup\:\: (R_2\cap I)\quad \text{and} \quad (R_1\cap I) \:\:\cap\:\: (R_2\cap I)=\emptyset.
\end{equation}
Also, it follows from~(i), (ii),~\eqref{eq:h},~\eqref{eq:prooflem1}, and the fact that the graph is simple
that 
\begin{equation}\label{eq:partition}
\begin{aligned}
L&=L(R_1\cap I)\:\:\cup\:\: L(R_1\setminus I)\:\:\cup\:\: L(R_2\cap I)\:\:\cup\:\: L(R_2\setminus I) \quad \text{and }\\
&L(R_1\cap I), \:\:\ L(R_1\setminus I), \:\:  L(R_2\cap I),\:\: L(R_2\setminus I)\quad \text{are pairwise disjoint,}
\end{aligned}
\end{equation}
i.e., the sets constitute a partition of $L$.

By~(i),~(ii),~\eqref{eq:prooflem1}, ~\eqref{eq:partition1} and~\eqref{eq:partition} we have \[\sum_{i \in R_1\cap I}a_i +\sum_{l\in L(R_1\setminus I)}h_l \in\{0,1\}^{n_1},\] 
hence \[\sum_{i \in R_1\cap I}a_i +\sum_{l\in L(R_1\setminus I)}h_l-\sum_{l\in L(R_1\cap I)}h_l\in\{0,1\}^{n_1}\] 
Similarly,
\[
\sum_{i \in R_2\cap I}a_i+\sum_{l\in L(R_2\setminus I)}h_l-\sum_{l\in L(R_2\cap I)}h_l\in\{0,1\}^{n_1}.
\]
Therefore,
\[\sum_{i \in R_1\cap I}a_i +\sum_{l\in L(R_1\setminus I)}h_l-\sum_{l\in L(R_1\cap I)}h_l -\left(\sum_{i \in R_2\cap I}a_i+\sum_{l\in L(R_2\setminus I)}h_l-\sum_{l\in L(R_2\cap I)}h_l\right)\in\{0,\pm1\}^{n_1},\]
i.e., the collection $I$ and $L$ of rows of $A$ and $H,$ respectively, is split into two parts such that the sum of the rows in one part minus the sum of the rows in the other part is a vector with entries in $\{0,\pm 1\}$.  
\end{proof}

Now we present a polynomial-time algorithm for \textsc{Rec}$(k,\nu,2)$ when $2\le k\le \nu$.

\begin{proof}[Proof of Theorem~\ref{thm:tract}(\ref{thm:tract:iii})]
We claim that the problem can be formulated as that of finding an integer solution to $\{x\::\: Mx \leq z, \: x\geq 0\}$ with a totally unimodular matrix~$M$ and integral right-hand side~$z,$ showing that the problem is solvable in polynomial time; see \cite{gritzmann97}, \cite[Thm.~16.2]{schrijver}.

Assembling the variables $\xi_{i,j},$ for $(i,j)\in  [m]\times[n]$ into an $mn$-dimensional vector and rephrasing the row and column sum constraints in matrix form  $Ax= (r^T,c^T)^T,$ where $r$ contains the~$r_j$'s and~$c$ contains the~$c_i$'s,~$A$ is the node-edge incidence matrix of a bipartite graph, and hence totally unimodular. 

As the pattern constraints ensure that no row of a block contains two 1's, and since $v(i,j)=0$ or $v(i,j)\geq k,$ we can (equivalently) replace the block and pattern constraints by the \emph{box constraints}
\begin{equation}\label{W2o}
\sum_{\mathclap{(p,q)\in W(i,j+l)}}\xi_{p,q}\leq \min\{1,v(i,j)\}, \qquad (i,j) \in C(m,n,k), \quad l\in [k-1]_0,
\end{equation} where $W(i,j+l):= B_k(i,j)\cap([m]\times\{j+l\}).$ 

We can rephrase the box constraints in matrix form $Hx\leq v,$ with a binary matrix~$H$ and a vector~$v$ containing the respective right-hand sides. 

Since each box $W(i,j+l)$ is contained in a row of $[m]\times [n]$ the condition in Lemma~\ref{prop:totallyunimodular}\eqref{lem1:prop2} is satisfied. Also, any two of these boxes are disjoint since the blocks are disjoint. Hence the condition in Lemma~\ref{prop:totallyunimodular}\eqref{lem1:prop3} is satisfied. Therefore, by Lemma~\ref{prop:totallyunimodular},  the matrix
\[
\left(\begin{array}{l}A\\H\end{array}\right)
\]
is also totally unimodular. Appending $-A$ and the identity matrix $E$ yields again a totally unimodular matrix, hence 
\[
M:=
\left(\begin{array}{r}A\\H\\-A\\E\end{array}\right)
\] is totally unimodular. Since the corresponding right-hand side vector~$z=(r^T,c^T,v^T,-r^T,-c^T,\1^T)^T$ is integral, the instance can be solved by linear programming.
\end{proof}

\section{Intractability results: $\NP$-hardness}\label{sec-intractability}

This section contains the proof of Theorem \ref{thm:NP}. We begin with the $\mathbb{N}\mathbb{P}$-hardness of $\textsc{Rec}(k,1,1)$ for $k\geq2$.

\begin{proof}[Proof of Theorem~\ref{thm:NP}(\ref{thm:NP:i})]
It suffices to show the result for $k=2$ since the $\mathbb{N}\mathbb{P}$-hardness for larger~$k$ can be inferred from that for $k=2$ by setting the row and column sums $r_{j+l}$ and $c_{i+l}$ to zero  for every $(i,j)\in C(m,n,k)$ and $l \in \{3,\dots,k-1\};$ see \cite{agsuperresolution}.

So, let $k=2.$ We use a transformation from the following problem \textsc{3-color tomography}, which is $\mathbb{N}\mathbb{P}$-hard by~\cite{duerr-Guinez-matamala-12}. (Note that the third color is just ``blank.'')

\begin{center}
\begin{minipage}{0.95\textwidth}
\textsc{3-color tomography}\\[-3ex]

\hspace*{2ex}\begin{minipage}{0.5\textwidth}
\begin{alignat*}{6}
&\textnormal{Instance:}\quad&& \omit \rlap{$\displaystyle m,n \in \mathbb{N},$ $r^{(a)}_1,\dots,r^{(a)}_n, c^{(a)}_1,\dots,c^{(a)}_m\in \mathbb{N}_0$ for $a\in[2].$}\\[1.2ex]
&\textnormal{Task:}\quad&& \omit \rlap{ Find $\xi_{p,q}^{(a)}\in\{0,1\},\:\:$ $(p,q)\in[m]\times[n],$ $a \in [2],$ with }\\
\displaystyle &&&\sum_{p\in[m]} \xi^{(a)}_{p,q}&&=r_q^{(a)},\qquad &&  q\in[n], a\in[2],&&\textnormal{(row sums)} \\ 
\displaystyle &&&\sum_{q\in[n]} \xi^{(a)}_{p,q}&&=c_p^{(a)},&&  p\in[m], a\in[2],&&\textnormal{(columns sums)} \\ 
\displaystyle &&&\xi_{p,q}^{(1)}+\xi_{p,q}^{(2)}&&\leq1, &&  (p,q)\in [m]\times[n]\qquad&&\textnormal{(disjointness condition),}\\[1.2ex]
&&& \omit \rlap{or decide that no such solution exists.}
\end{alignat*}
\end{minipage}
\end{minipage}
\end{center}

Let $\mathcal{I}=(m,n, r^{(1)}_1,\dots,r^{(1)}_n, c^{(1)}_1,\dots,c^{(1)}_m, r^{(2)}_1,\dots,r^{(2)}_n, c^{(2)}_1,\dots,c^{(2)}_m)$ denote an instance of \textsc{3-color tomography}. We are setting up an instance~$\mathcal{I}'$ of our reconstruction problem by defining the row and column sums $r_1,\dots,r_{2n}$ and $c_1,\dots,c_{2m}$, respectively. 

The block and pattern constraints ensure that there are only three possibilities for setting the 1's in each block. They are shown in Figure~\ref{fig:3colors}.
\begin{figure}[htb] 
\centering
\includegraphics[width=0.3\textwidth]{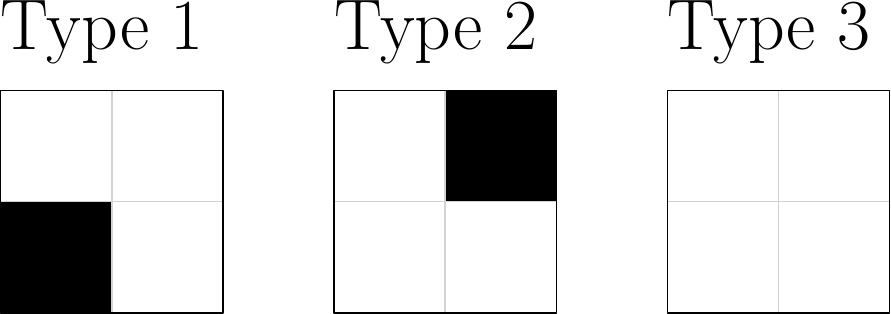}
\caption{The three possible types of blocks in solutions of $\textsc{Rec}(2,1,1)$ instances.}\label{fig:3colors}
\end{figure}

Note that the types can be viewed as representing three colors, which are counted by different row and column sums. The row and column sums with odd indices count type~1 blocks, the other sums count type~2 blocks. 

Setting
\begin{alignat*}{4}
\displaystyle r_{2q+a-2}&:=r^{(a)}_q,\quad && q\in[n],\:\: &a\in[2], \\
\displaystyle c_{2p+a-2}&:=c^{(a)}_p,\quad && p\in[m],\:\: &a\in[2],
\end{alignat*} and using  the correspondences 
\begin{alignat*}{4}
\displaystyle {\xi_{p,q}^{(a)}}^*=1 \quad &&\Leftrightarrow \quad (\xi_{2p-1,2q-1}^*,\xi_{2p-1,2q}^*,\xi_{2p,2q-1}^*,\xi_{2p,2q}^*) \textnormal{ is of type~a}, &&\qquad \textnormal{for } a \in[2],\\
\displaystyle {\xi_{p,q}^{(1)}}^*={\xi_{p,q}^{(2)}}^*=0 \quad &&\Leftrightarrow \quad (\xi_{2p-1,2q-1}^*,\xi_{2p-1,2q}^*,\xi_{2p,2q-1}^*,\xi_{2p,2q}^*) \textnormal{ is of type~3},&&
\end{alignat*} we conclude that $\mathcal{I}$ admits a solution if, and only if, $\mathcal{I}'$ admits a solution.
\end{proof}

Next we turn to the $\NP$-hardness of \textsc{Rec}$(k,2,0)$ for $k\geq 2$. Again it suffices to prove the result for~$k=2.$ 

\begin{proof}[Proof of Theorem~\ref{thm:NP}(\ref{thm:NP:ii})] 
The general structure of the proof follows that of the $\mathbb{N}\mathbb{P}$-hardness result under data uncertainty given in \cite{agsuperresolution}. However, some adaptations and additions are required as will be detailed below. 
 
As in~\cite{agsuperresolution} we give a transformation from the $\mathbb{N}\mathbb{P}$-hard problem \textsc{1-In-3-SAT} \cite{cook1971}, which asks for a satisfying truth assignment that sets exactly one literal true in each clause of a given Boolean formula in conjunctive normal form where all clauses contain three literals (involving three different variables). 

For a given instance of \textsc{1-In-3-SAT} a ``circuit board'' is constructed that contains an initializer, several connectors, and clause chips.  The general structure of the circuit board from the proof of~\cite{agsuperresolution} and its modification, which contains additional connectors and initializers,  are outlined in Figure~\ref{fig:P2layoutleq}.

\begin{figure}[htb] 
\centering
\subfigure[]{\includegraphics[width=0.4\textwidth]{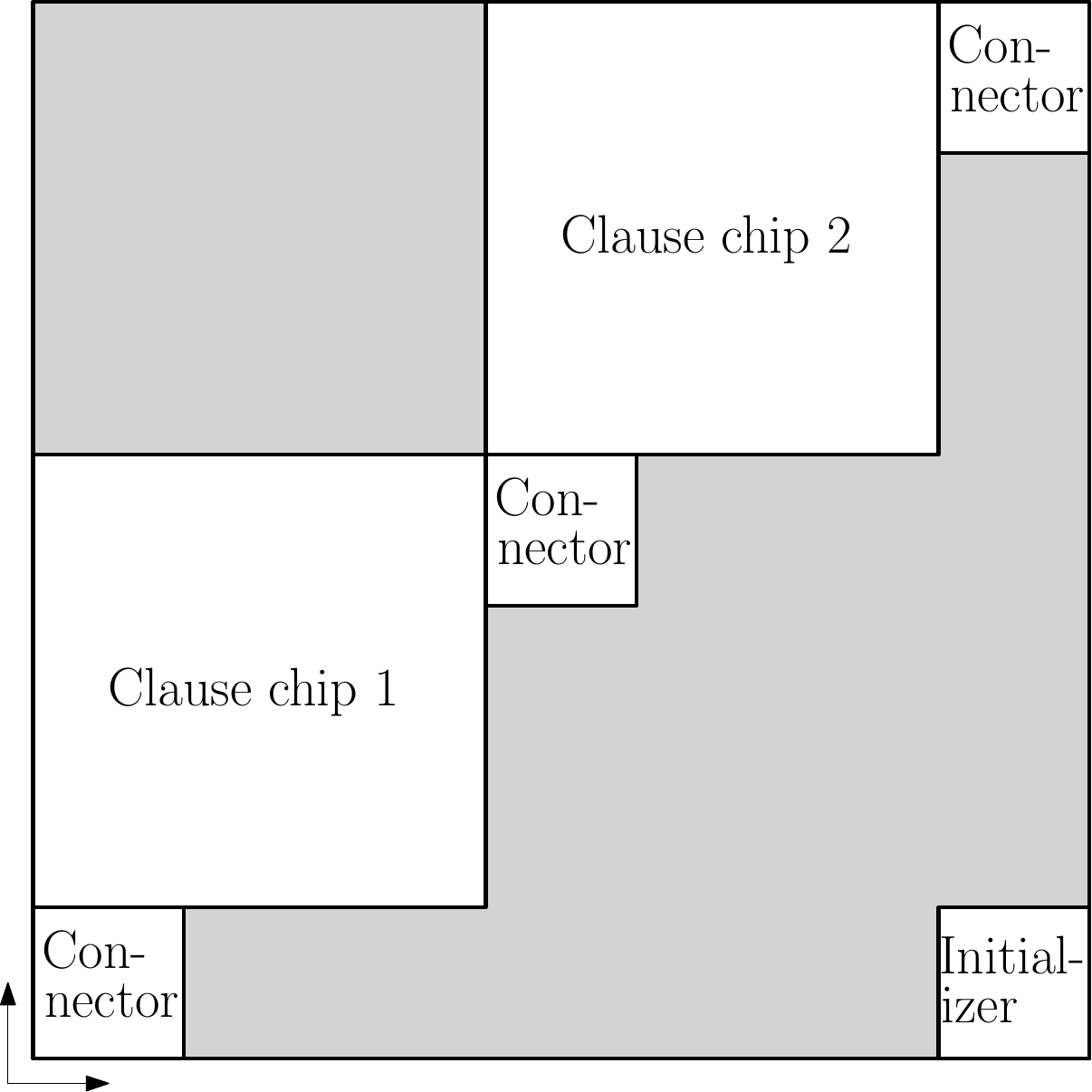}}\hspace*{4ex}
\subfigure[]{\includegraphics[width=0.4\textwidth]{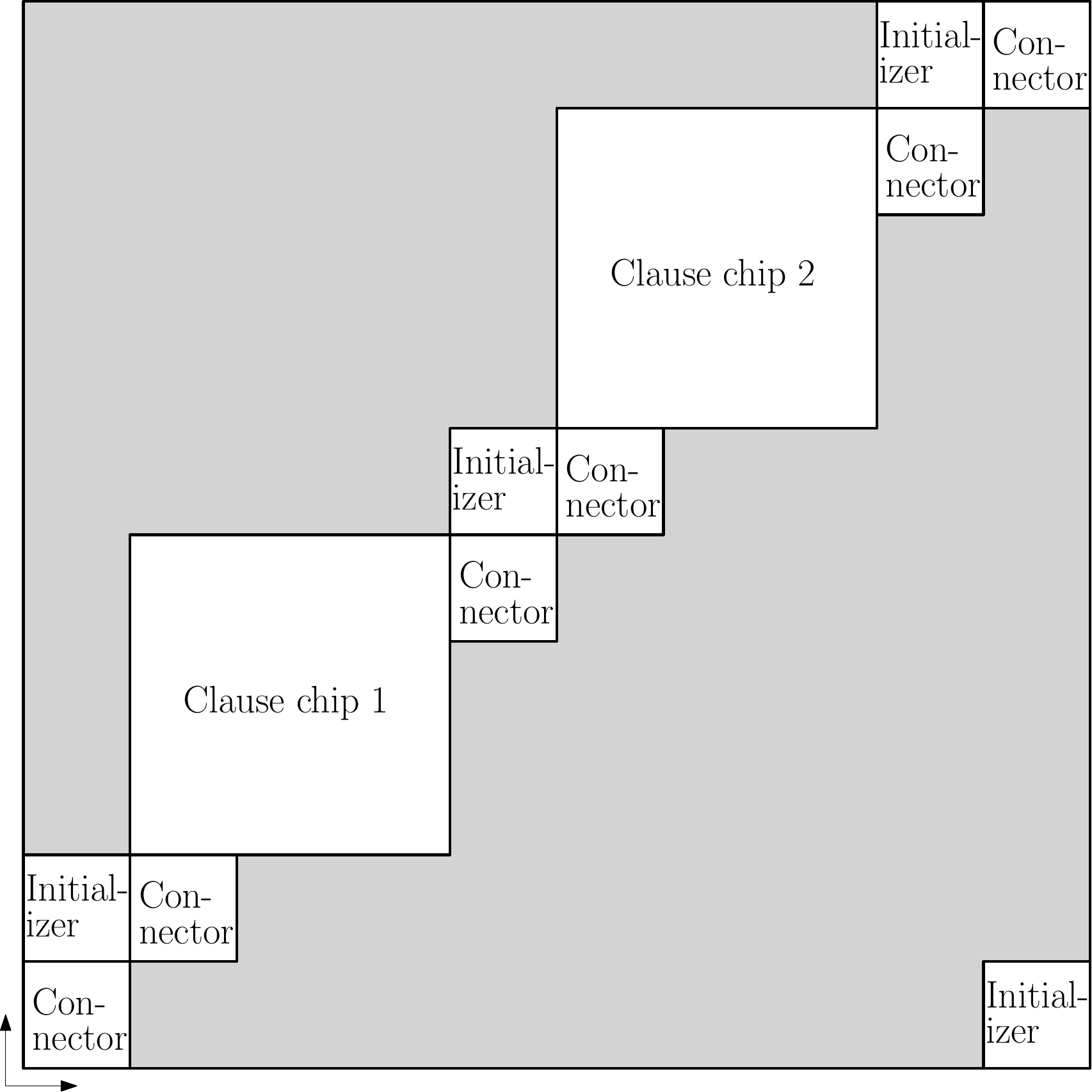}}
\caption{The general layout of (a) the original circuit board and (b) the modified circuit board, here for two clauses.}\label{fig:P2layoutleq}
\end{figure}

The initializers contain, for every variable $\tau_t,$ $t\in[T],$ so-called \emph{$\tau_t$-chips}, while the connectors contain \emph{$\lnot\tau_t$-chips.} The clause chips are more complex as they consist of two \emph{collectors}, two \emph{verifiers}, and a \emph{transmitter}. 
The initializer holds a truth assignment for the variables $\tau_1,\dots,\tau_T$ of the given instance~$\mathcal{I}$ of \textsc{1-In-3-SAT}. The Boolean values \textsc{True} and \textsc{False} are encoded in $\tau_t$-chips, $t\in[T],$ by the type~1 and~2 blocks, respectively, that are shown in Figure~\ref{fig:P2types}.

\begin{figure}[htb] 
\centering
\includegraphics[width=0.17\textwidth]{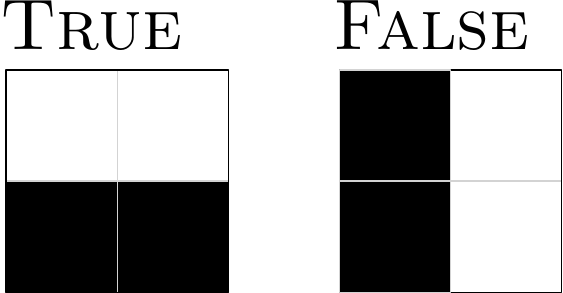}
\caption{Boolean values in $\tau_t$-chips.}\label{fig:P2types}
\end{figure}

The proof in~\cite{agsuperresolution} shows how the truth assignment is transmitted through the circuit board, and how the verifier chips indeed check the feasibility of the given \textsc{1-In-3-SAT} instance. Figure~\ref{fig:P2example2} depicts a specific problem instance and a corresponding solution for the original construction of~\cite{agsuperresolution} .

\begin{figure}[htb] 
\centering
\subfigure[]{\includegraphics[width=0.47\textwidth]{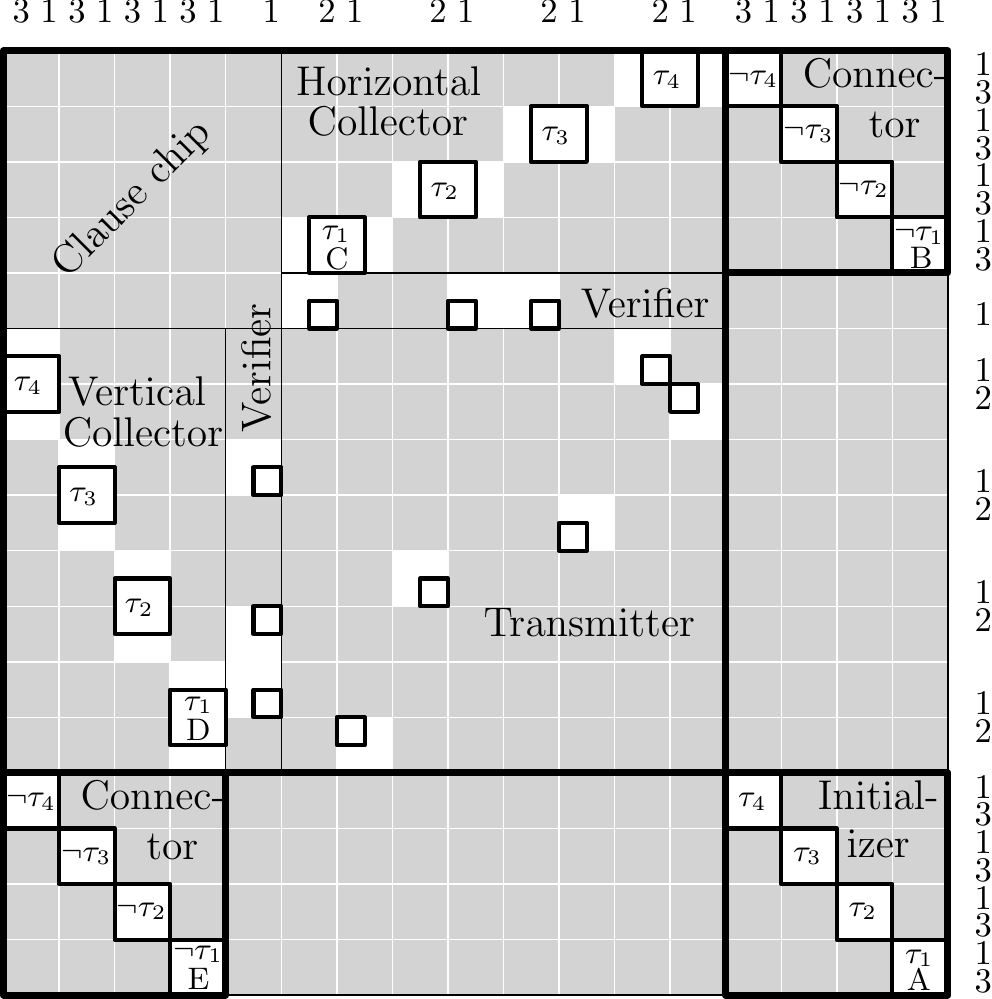}}\hfill
\subfigure[]{\includegraphics[width=0.47\textwidth]{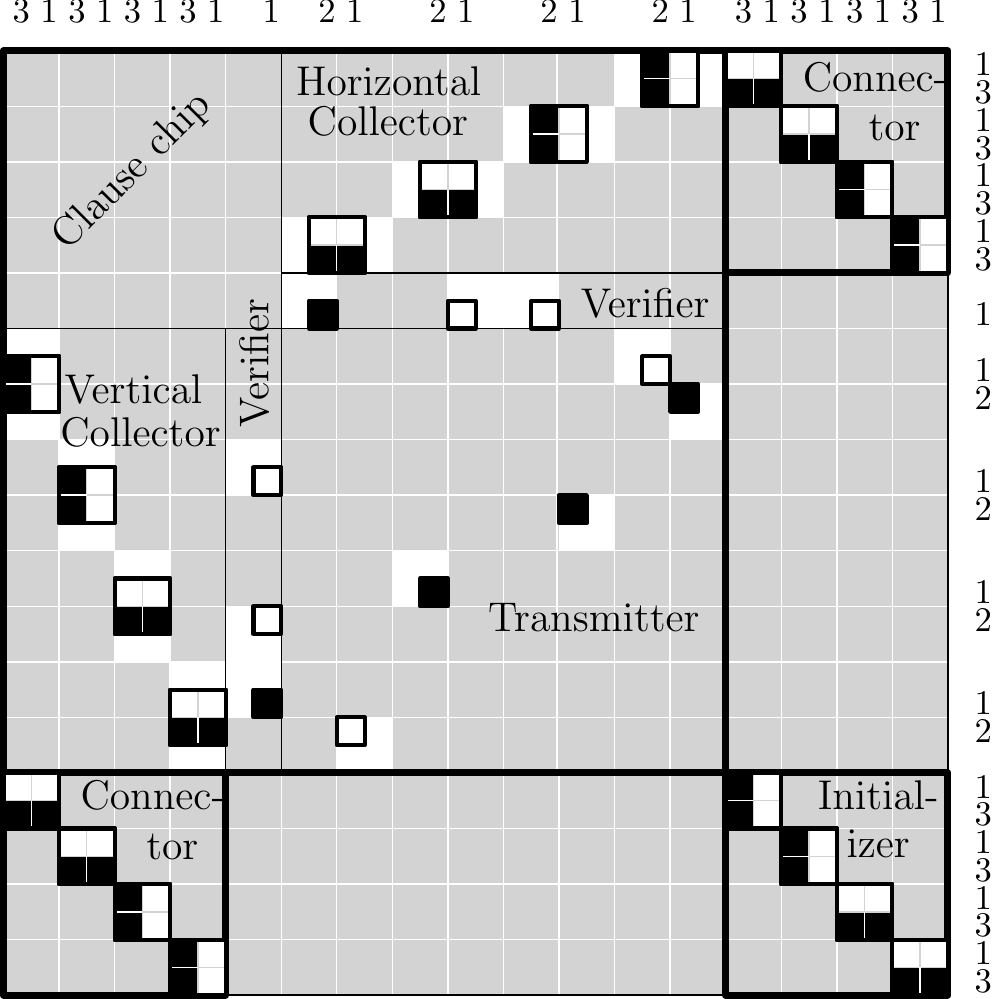}}
\caption{(\cite{agsuperresolution}) Original transformation from \textsc{1-In-3-SAT} for the instance  that involves the single clause $\tau_1\vee \lnot\tau_2 \vee \tau_3.$  (a) The circuit board. By setting to zero suitable row and column sums it is ensured that the non-zero components $\xi^*_{p,q}$ of a solution are only possible in the bold-framed boxes within the white blocks in the clause chip, connectors, and the initializer.  (b) A  solution $x^*$ (non-zero components~$\xi^*_{p,q}$ are depicted as black pixels), representing the Boolean solution $(\tau^*_1,\tau^*_2,\tau^*_3,\tau^*_4)=(\textsc{True},\textsc{True},\textsc{False},\textsc{False}).$  }\label{fig:P2example2}
\end{figure}

The  proof in~\cite{agsuperresolution} employs three different types of block constraints. For every block $B_2(i,j),$ $(i,j)\in C(m,n,2),$ there is one of the following constraints:
\begin{alignat*}{4}
\sum_{(p,q)\in B_2(i,j)}\xi_{p,q}&= 0, \qquad &&\textnormal{($(=,0)$-block constraint)}\\
\sum_{(p,q)\in B_2(i,j)}\xi_{p,q}&= 2,\qquad &&\textnormal{($(=,2)$-block constraint)}\\
\sum_{(p,q)\in B_2(i,j)}\xi_{p,q}&\leq 2, \qquad &&\textnormal{($(\approx,1)$-block constraint)};
\end{alignat*} the names in parenthesis are those from \cite{agsuperresolution} (even though~$(\leq,2)$ may seem more natural here in the third case).

As for binary variables, the constraint $\sum_{(p,q)\in B_2(i,j)}\xi_{p,q}= 0$ is equivalent to $\sum_{(p,q)\in B_2(i,j)}\xi_{p,q}\leq 0,$ it suffices to adapt the construction in such a way that no $(=,2)$-block constraints are required. In fact, the $(=,2)$-block constraints in the original construction are only employed for the $\lnot\tau_t$-chips, $t\in[T],$ contained in the connectors. 

For each $\lnot\tau_t$-chip, $t\in[T],$ in a connector, we replace the $(=,2)$-block constraints by a $(\approx,1)$-block constraint. Then,  between any two clause chips we insert an additional copy of a connector and an initializer each. These copies are placed as indicated in Figure~\ref{fig:P2layoutleq}(b). The additional row and column sums are prescribed again to values~$1$ and~$3$ in alternation.  
The new copies ensure that each $\lnot\tau_t$-chip, $t\in[T],$ is contained in a horizontal or vertical strip that contains a $\tau_t$-chip of an initializer. The corresponding row or column sums have values~$1$ and~$3,$ and since we have $(\approx,1)$-block constraints for both chips, we can distribute the $1+3=4$ ones only in such a way that each of the two chips contains exactly two 1's. In other words, each $\lnot\tau_t$-chip, $t\in[T],$ is required to contain exactly two 1's. Figure~\ref{fig:change} shows the solution of our correspondingly adapted circuit board from~Figure~\ref{fig:P2example2}.

\begin{figure}[htb] 
\centering
\includegraphics[width=0.6\textwidth]{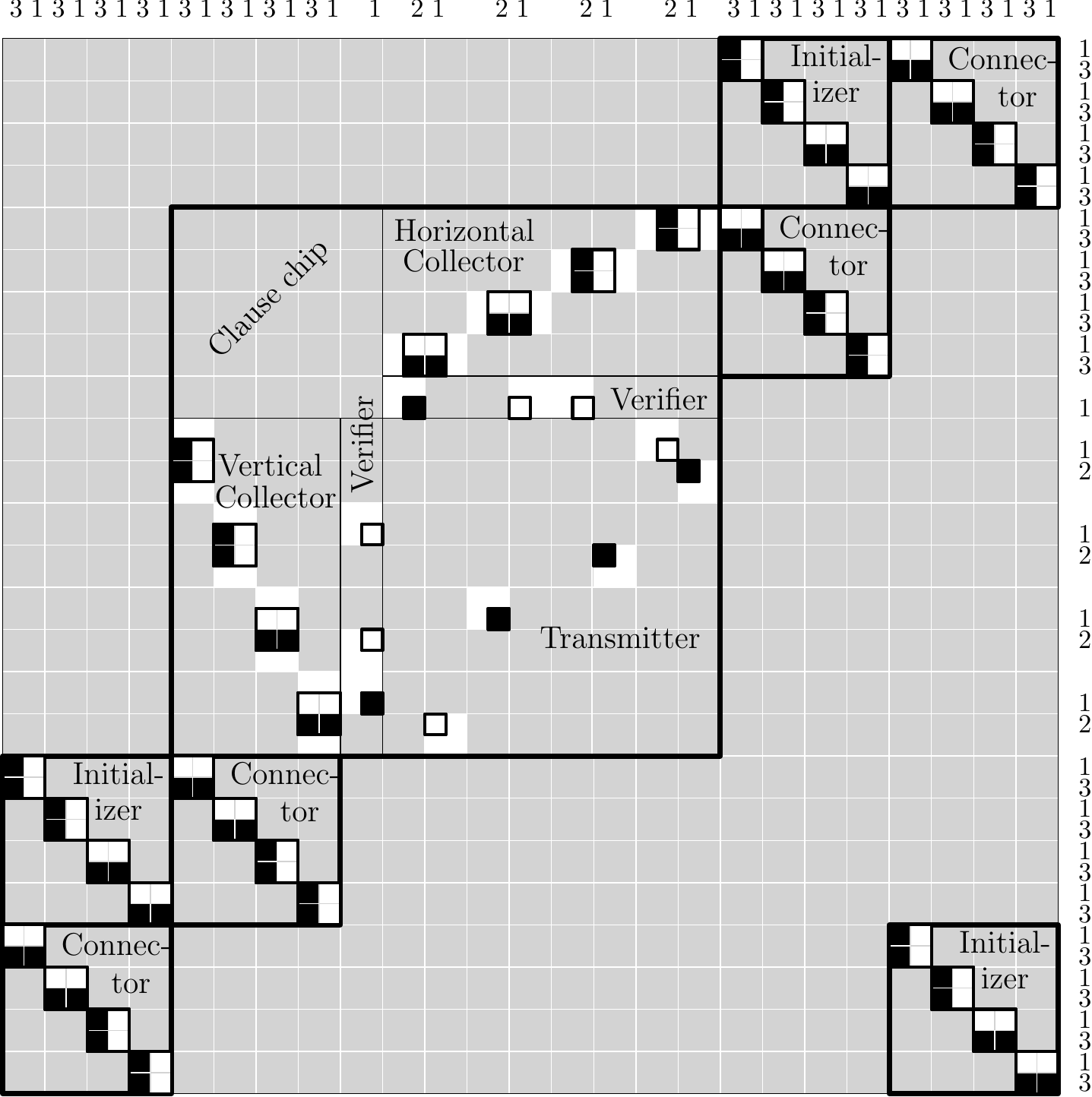}\vspace*{1ex}
\caption{Transformation from \textsc{1-In-3-SAT} for the instance  that involves the single clause $\tau_1\vee \lnot\tau_2 \vee \tau_3.$  A  solution $x^*$ is shown (non-zero components~$\xi^*_{p,q}$ are depicted as black pixels), representing the Boolean solution $(\tau^*_1,\tau^*_2,\tau^*_3,\tau^*_4)=(\textsc{True},\textsc{True},\textsc{False},\textsc{False}).$  }\label{fig:change}
\end{figure}

With this adaptation the proof from~\cite{agsuperresolution} carries over. As a service to the reader we explain how the truth assignment is transmitted in our particular example; see Figure~\ref{fig:transmission}. The formal proof that, in particular, this transmission and, in general, the complete transformation works as needed, follows exactly as in~\cite{agsuperresolution}.

\begin{figure}[htb] 
\centering
\includegraphics[width=0.6\textwidth]{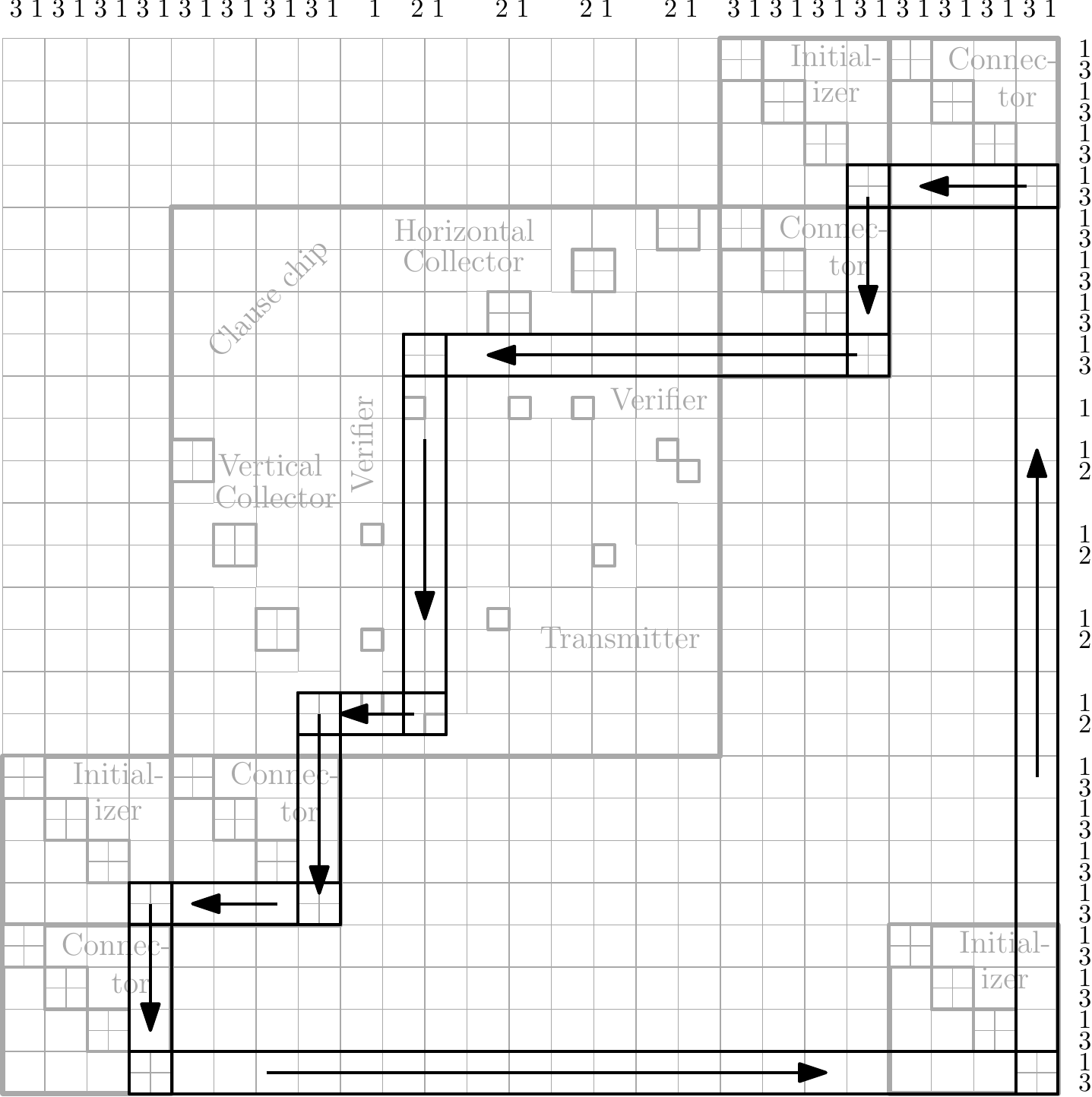}\vspace*{1ex}
\caption{Transmission of the Boolean values for~$\tau_1$ (for the example shown in~Figure~\ref{fig:change}). }\label{fig:transmission}
\end{figure}

Suppose the $\tau_1$-chip of the initializer in the lower right corner of the circuit board shown in Figure~\ref{fig:transmission} is of type~1. 
The truth assignment is vertically transmitted from the initializer to a connector, which negates the Boolean value as the X-rays are set to~$1$ and~$3$ in alternation. Then, this Boolean value is horizontally transmitted to an initializer (again negated, so that we have the original values back), then vertically to a connector (again negated). From here the Boolean value is transmitted horizontally to the first clause chip, where it enters and exits, yet again negated, through a collector. In the part between the two collectors, the Boolean values for the literals appearing in the first clause take a path through the verifiers that ensure that exactly one literal is set to true; the values for the other literals are just transmitted (again with negation). This process continues in a similar way for the remaining clauses. We end up with feasible solution for our instance of \textsc{Rec}$(k,2,0)$ if, and only if, there is a satisfying truth assignment for our given instance of \textsc{1-In-3-SAT}.
\end{proof}

\section{Final Remarks}\label{sect:finalremarks}
Let us first point out that the question of inferring information about an otherwise unknown object from observations ``under the microscope,'' i.e., through certain windows, is fundamental for a wide range of applications. For corresponding results on the scanning of binary or integer matrices; see e.g.~\cite{frosinimicroscope07, frosininivatrinaldi08}. In our terminology, \cite{frosinimicroscope07} studies the task of reconstructing binary matrices from their number of~$1$'s in each windows of a fixed size. 

We will close by introducing a more general problem and stating some 
implications of our main Theorems \ref{thm:tract} and \ref{thm:NP}.

As before, $m,n,k\in\mathbb{N}$. Further, for $V_\leq, V_\geq, V_=\subseteq \mathbb{N}_0$ we set 
\[R(V_\leq,V_\geq,V_=):=(\{\leq\}\times V_\leq)\cup(\{\geq\}\times V_\geq)\cup(\{=\}\times V_=).\]
Also, let for $(i,j)\in\mathbb{N}_0^2$ 
\[ W_{k}(i,j):=(i,j)+{[k-1]^2_0},
\] 
and define for $L\subseteq \mathbb{N}^2$ 
\[
C(m,n,k,L):=\left\{(p,q)\in L \::\: W_{k}(p,q)\subseteq [m]\times[n] \right\}.
\] 

For $L\subseteq \mathbb{N}^2,$ $(i,j)\in C(m,n,k,L),$ and $x=(\xi_{i,j})_{i\in[m],j\in[n]}$ 
we set 
\[\textnormal{pat}_{k}(x,i,j):=\{(p,q)\in W_{k}(i,j)\::\: \xi_{p,q}\neq0\}-(i,j).\]

In addition to the patterns $P(k,0),P(k,1)$, and $P(k,2)$ we consider, for $k\geq 2$,
a fourth one 
\[P(k,3):=\{M\in2^{[k-1]_0^2} \::\: |M\cap\left([k-1]_0\times\{j\}\right)|\geq k-1 \textnormal{ for all } j\in [k]_0\}. \]
Note that $P(k,3)$ can be viewed as the colored-inverted version of $P(k,2).$

Now, let $V_\leq,V_\geq,V_=\subseteq \mathbb{N}_0,$ $t\in[3]_0,$ and $L\subseteq \mathbb{N}^2$. Then we define the following problem.

\begin{center}
\begin{minipage}{0.95\textwidth}
$\textsc{WRec}(k,V_\leq,V_\geq,V_=,t,L)$\\[-3ex]

\hspace*{1ex}\begin{minipage}{0.98\textwidth}
\begin{alignat*}{6}
&\textnormal{Instance:}\quad&& \omit \rlap{$\displaystyle m,n \in k\mathbb{N},$}\\
&&&\omit \rlap{$\displaystyle r_1,\dots,r_n \in \mathbb{N}_0,$}&&&&&& \textnormal{(row sum measurem.)} \\
&&&\omit \rlap{$\displaystyle c_1,\dots,c_m \in \mathbb{N}_0,$}&&&&&& \textnormal{(col. sum measurem.)} \\
&&&\omit \rlap{$(\sim_{i,j},v_{i,j})\in R(V_\leq,V_\geq,V_=),$} &&&& (i,j)\in C(m,n,k,L), &&  \textnormal{(window measurem.)}\\[1.2ex]
&\textnormal{Task:}\quad&& \omit \rlap{Find  $\xi_{p,q}\in\{0,1\},$ \:$(p,q)\in[m]\times[n],$ with }\\
\displaystyle &&&\sum_{\mathclap{p\in[m]}} \xi_{p,q}&&=r_q,&&   q\in[n], &&\textnormal{(row sums)} \\ 
\displaystyle &&&\sum_{q\in[n]} \xi_{p,q}&&=c_p,&&  p \in [m], &&\textnormal{(column sums)}\\  
\displaystyle &&&\sum_{\mathclap{(p,q)\in W_{k}(i,j)}} \xi_{p,q}&& \sim_{i,j} v(i,j),&&  (i,j)\in C(m,n,k,L), \quad &&\textnormal{(window constraints)}\\
\displaystyle &&&\textnormal{pat}_{k}(x,i,j)&&\in P(k,t), \qquad \qquad \qquad &&  (i,j)\in C(m,n,k,L),\qquad  &&\textnormal{(pattern constraints),}\\[1.2ex]
&&& \omit \rlap{or decide that no such solution exists.}
\end{alignat*}
\end{minipage}
\end{minipage}
\end{center}

Tables~\ref{tableP01} and~\ref{tableNP01} below list tractability and intractability results for 
$\textsc{WRec}(k,V_\leq,V_\geq,V_=,t,L)$, respectively, which are simple corollaries to our main
Theorems \ref{thm:tract} and \ref{thm:NP}. In fact, they are either just reinterpretations of these results or rely, in addition, on one or two of the following simple transformation principles (T1), (T2), or (T3).

(T1) is the trivial observation that rather than reconstructing the ones in an image, we may reconstruct the zeros instead. More precisely, this process of ``color inversion'' works as follows. Let $\alpha$ denote the number of positions in the window $W_{k}(0,0)$.
Then, we associate a given instance $\mathcal{I}$ of 
$\textsc{WRec}(k,V_\leq,V_\geq,V_=,t,L)$ with
the instance $\mathcal{I}'$ of $\textsc{WRec}(k,V'_\leq,V'_\geq,V'_=,t,L)$ where $V'_\leq:=\{\alpha-l\::\:l\in V_\leq\},$  $V'_\geq:=\{\alpha-l\::\:l\in V_\geq\},$ and $V'_=:=\{\alpha-l\::\:l\in V_=\}$,
$r_{j}':=m-r_{j},$ $j\in[n],$ $c_{i}':=n-c_{i},$ $i\in[m],$ 
$P'(k,t):=\left\{W_{k}(0,0)\setminus M\::\: M\in P(k,(t)\right\},$ and, for all $(i,j)\in C$,
$(\sim_{i,j},v(i,j))$ is replaced by 
\[
\left\{\begin{array}{lll} 
(\geq,\alpha-v(i,j))&:& \sim_{i,j}\in\{\leq\},\\
(\leq,\alpha-v(i,j))&:& \sim_{i,j}\in\{\geq\},\\
(=,\alpha-v(i,j))&:& \sim_{i,j}\in\{=\}.
 \end{array}\right. 
\]  
Clearly, the problems $\textsc{WRec}(k,V_\leq,V_\geq,V_=,t,L)$  and $\textsc{WRec}(k,V'_\leq,V'_\geq,V'_=,t,L)$  lie in the same complexity class.

The transformation (T2) adds empty rows and columns to an 
instance $\mathcal{I}$ of $\textsc{WRec}(2,V_\leq,V_\geq,V_=,0,L)$ 
to extend $\mathbb{N}\mathbb{P}$-hardness results for $k=2$ to higher values of $k$. More precisely, let $m':=mk/2,$ $n':=nk/2,$ and set for $(i,j)\in ([m]\times[n])\cap C(m,n,2)$ and $l\in[k]$
\begin{align*}
r'_{\frac{k}{2}(j-1)+l}&:=\left\{\begin{array}{lll}r_{j+l}&:&l\in[2],\\ 0&:&\textnormal{otherwise,}\end{array}\right.\\
c'_{\frac{k}{2}(i-1)+l}&:=\left\{\begin{array}{lll}c_{i+l}&:&l\in[2],\\ 0&:&\textnormal{otherwise,}\end{array}\right.\\
(\sim'_{i,j},v'(k(i-1)/2+1,k(j-1)/2+1)))&:=(\sim_{i,j},v(i,j)).
\end{align*} 
This defines an instance $\mathcal{I}'$ of $\textsc{WRec}(k,V_\leq,V_\geq,V_=,0,L)$ with row sums $r'_1,\dots,r'_{n'},$ column sums $c'_1,\dots,c'_{m'},$ and block constraints $(\sim'_{i,j},v'(i,j)),$ $(i,j)\in C(m',n',3,k,L).$  Clearly, the instance~$\mathcal{I}$  admits a solution if, and only if,~$\mathcal{I}'$  admits a solution (by filling/extracting the~$2\times2$-blocks of~$\mathcal{I}$ into/from the~$k\times k$-blocks of $\mathcal{I}'$). As this is a polynomial-time transformation, $\mathbb{N}\mathbb{P}$-hardness of $\textsc{WRec}(2,V_\leq,V_\geq,V_=,0,L)$ implies $\mathbb{N}\mathbb{P}$-hardness of $\textsc{WRec}(k,V_\leq,V_\geq,V_=,0,L).$ 

In the final transformation (T3) additional rows and columns are filled by ones. For a  given instance $\mathcal{I}$ of $\textsc{WRec}(2,V_\leq,V_\geq,V_=,0,L),$ and  $k\geq2$, let $m':=mk/2$, $n':=nk/2$, and set for $(i,j)\in ([m]\times[n])\cap C(m,n,2)$ and $l\in[k]$
\begin{align*}
r'_{\frac{k}{2}(j-1)+l}&:=\left\{\begin{array}{lll}r_{j+l}&:&l\in[2],\\ m'&:&\textnormal{otherwise,}\end{array}\right.\\
c'_{\frac{k}{2}(i-1)+l}&:=\left\{\begin{array}{lll}c_{i+l}&:&l\in[2],\\ n'&:&\textnormal{otherwise,}\end{array}\right.\\
(\sim'_{i,j},v'(k(i-1)/2+1,k(j-1)/2+1)))&:=(\sim_{i,j},k^2+v(i,j)-4).
\end{align*} 
This defines an instance $\mathcal{I}'$ of $\textsc{WRec}(k,V'_\leq,V'_\geq,V'_=,0,L)$ with row sums $r'_1,\dots,r'_{n'},$ column sums $c'_1,\dots,c'_{m'},$  block constraints $(\sim'_{i,j},v'(i,j)),$ $(i,j)\in C(m',n',3,k,L),$  $V'_\leq:=\{k^2+v-4:v\in V_\leq\},$ $V'_\geq:=\{k^2+v-4:v\in V_\geq\},$ and $V'_=:=\{k^2+v-4:v\in V_=\}.$ Clearly, the instance~$\mathcal{I}$  admits a solution if, and only if,~$\mathcal{I}'$  admits a solution (by filling/extracting the~$2\times2$-blocks of~$\mathcal{I}$ into/from the~$k\times k$-blocks of $\mathcal{I}'$). As this is a polynomial-time transformation, $\mathbb{N}\mathbb{P}$-hardness of $\textsc{WRec}(2,V_\leq,V_\geq,V_=,0,L)$ implies $\mathbb{N}\mathbb{P}$-hardness of $\textsc{WRec}(k,V'_\leq,V'_\geq,V'_=,0,L).$

{\footnotesize
\begin{table}[htb]
\begin{tabular}{c@{\hskip 3ex}c@{\hskip 3ex}c@{\hskip 3ex}c@{\hskip 3ex}c@{\hskip 3ex}c@{\hskip 3ex}l@{\hskip 3ex}}\toprule
$k$&$V_{\leq}$&$V_{\geq}$&$V_=$&$t$&$L$&Reference\\\midrule
1&$\mathbb{N}_0$&$\emptyset$&$\emptyset$&$[2]_0$&$\mathbb{N}^2$&Thm.~\ref{thm:tract}(\ref{thm:tract:i})\\
2&$\{0\}$&$\{4\}$&$[4]_0$&0&$(2\mathbb{N}_0+1)^2$&\cite{agsuperresolution}\\
$\geq2$&$\{0,1\}$&$\{k^2\}$&$\{0,k^2\}$&$0$&$(k\mathbb{N}_0+1)^2$&Thm.~\ref{thm:tract}(\ref{thm:tract:ii})\\
$\geq2$&$\{0\}$&$\{k^2-1,k^2\}$&$\{0,k^2\}$&$0$&$(k\mathbb{N}_0+1)^2$&Thm.~\ref{thm:tract}(\ref{thm:tract:ii}) \& (T1)\\
$\geq2$&$\{0\}\cup\{\nu:\nu\geq k\}$&$\emptyset$&$\{0\}$&$2$&$(k\mathbb{N}_0+1)^2$&Thm.~\ref{thm:tract}(\ref{thm:tract:iii})\\
$\geq2$&$\emptyset$&$\{k^2\}\cup[k(k-1)]_0$&$\{k^2\}$&$3$&$(k\mathbb{N}_0+1)^2$&Thm.~\ref{thm:tract}(\ref{thm:tract:iii}) \& (T1)\\
\bottomrule\\[-.1ex]
\end{tabular}
\caption{Polynomial-time solvable problems $\textsc{WRec}(k,V_\leq,V_\geq,V_=,t,L).$} \label{tableP01}
\end{table}
}

{\footnotesize
\begin{table}[htb]
\begin{tabular}{c@{\hskip 4ex}c@{\hskip 4ex}c@{\hskip 4ex}c@{\hskip 4ex}c@{\hskip 4ex}c@{\hskip 4ex}l@{\hskip 4ex}}\toprule
$k$&$V_{\leq}$&$V_{\geq}$&$V_=$&$t$&$L$&Reference\\\midrule
$\geq2$&$\emptyset$&$\{0\}$&$\{0,1,2\}$&0&$\mathbb{N}^2$&\cite{agsuperresolution}
\refstepcounter{rownumber}\label{tableNP01:14}\\
$\geq2$&$\{0,1\}$&$\emptyset$&$\emptyset$&1&$(k\mathbb{N}_0+1)^2$&Thm.~\ref{thm:NP}(\ref{thm:NP:i})
\refstepcounter{rownumber}\\
$\geq2$&$\{0,2\}$&$\emptyset$&$\emptyset$&0&$(k\mathbb{N}_0+1)^2$&Thm.~\ref{thm:NP}(\ref{thm:NP:ii})
\refstepcounter{rownumber}\label{tableNP01:17}\\
$\geq2$&$\{k^2-4,k^2-2\}$&$\emptyset$&$\emptyset$&0&$(k\mathbb{N}_0+1)^2$&Thm.~\ref{thm:NP}(\ref{thm:NP:ii}) \& (T3)\\
$\geq2$&$\emptyset$&$\{k^2-2,k^2\}$&$\emptyset$&0&$(k\mathbb{N}_0+1)^2$&Thm.~\ref{thm:NP}(\ref{thm:NP:ii}) \& (T1)
\refstepcounter{rownumber}\\
$\geq2$&$\emptyset$&$\{2,4\}$&$\emptyset$&0&$(k\mathbb{N}_0+1)^2$&Thm.~\ref{thm:NP}(\ref{thm:NP:ii}) \& (T1) \& (T2)
\refstepcounter{rownumber}\\
\bottomrule\\[-.1ex]
\end{tabular}
\caption{$\mathbb{N}\mathbb{P}$-hard problems $\textsc{WRec}(k,V_\leq,V_\geq,V_=,t,L).$}\label{tableNP01}
\end{table}
}

Let us finally point out, that the Tables~\ref{tableP01} and~\ref{tableNP01} can be
largely extended. In particular, we can consider reconstruction problems for different types of windows which are even allowed to vary for each $(i,j)\in[m]\times[n]$. While some of the complexity results in this more general setting may seem rather marginal, other may be interesting for certain applications; see \cite{agdynamic}. As this would go far beyond the scope of the present paper we refrain, however, from carrying the results to the extremes here.


\end{document}